\documentclass[lettersize,journal]{IEEEtran}

\usepackage[disable]{todonotes} 
\definecolor{pastel_mint_green}{RGB}{198, 240, 230} 
\definecolor{purple_custom}{RGB}{65,35,142}
\presetkeys{todonotes}{backgroundcolor=pastel_mint_green,linecolor=pastel_mint_green}{}

\newcounter{todocounter}

\let\oldtodo\todo  
\renewcommand{\todo}[1]{%
    \refstepcounter{todocounter}%
    \oldtodo[inline]{\thetodocounter \ - #1}%
}

\usepackage{amsmath,amssymb,mathtools}
\usepackage{bm}            
\usepackage{mathrsfs}      

\usepackage{enumitem}

\usepackage{graphicx}

\usepackage{xcolor}
 
\usepackage{microtype}

\usepackage[colorlinks=true,linkcolor=blue,citecolor=blue,urlcolor=blue]{hyperref}

\usepackage[nameinlink,capitalise,noabbrev]{cleveref}

\newcommand{\R}{\mathbb{R}}

\newcommand{\Dx}{\nabla_x}
\newcommand{\Dt}{\partial_t}

\usepackage{amsmath,amsfonts}
\usepackage{algorithmic}
\usepackage{algorithm}
\usepackage{array}
\usepackage[caption=false,font=normalsize,labelfont=sf,textfont=sf]{subfig}
\usepackage{textcomp}
\usepackage{stfloats}
\usepackage{url}
\usepackage{verbatim}
\usepackage{cite}
\usepackage{booktabs}

\usepackage{amsthm}
\newtheorem{assumption}{Assumption}
\crefname{assumption}{Assumption}{Assumptions} 
\newtheorem{remark}{Remark}
\newtheorem{lemma}{Lemma}

\newtheorem{theorem}{Theorem}

\usepackage{xcolor}
\usepackage{soul}
\sethlcolor{yellow}

\newif\ifwithproofs
\withproofstrue   

\begin{document}

\title{Certifying Hamilton-Jacobi Reachability Learned via Reinforcement Learning}

\author{Prashant Solanki$^{1}$, Isabelle El-Hajj, Jasper J. van Beers, Erik-Jan van Kampen, Coen C. de Visser
\thanks{All authors are with the Control \& Operations department of the Faculty of Aerospace Engineering, Delft University of Technology, 2629 HS, Delft, the Netherlands.}
\thanks{$^{1}$Corresponding author}%
}

\markboth{Journal of \LaTeX\ Class Files,~Vol.~14, No.~8, August~2021}%
{Shell \MakeLowercase{\textit{et al.}}: A Sample Article Using IEEEtran.cls for IEEE Journals}

\IEEEpubid{0000--0000/00\$00.00~\copyright~2021 IEEE}

\maketitle

\begin{abstract}
We present a framework to \emph{certify} Hamilton--Jacobi (HJ) reachability learned by reinforcement learning (RL). Building on a discounted initial time \emph{travel-cost} formulation that makes small-step RL value iteration provably equivalent to a forward Hamilton--Jacobi (HJ) equation with damping, we convert certified learning errors into calibrated inner/outer enclosures of strict backward reachable tube. The core device is an additive-offset identity: if $W_\lambda$ solves the discounted travel-cost Hamilton--Jacobi--Bellman (HJB) equation, then $W_\varepsilon:=W_\lambda + \varepsilon$ solves the same PDE with a constant offset $\lambda\varepsilon$. This means that a uniform value error is \emph{exactly} equal to a constant HJB offset. We establish this uniform value error via two routes: (A) a Bellman operator-residual bound, and (B) a HJB PDE-slack bound. Our framework preserves HJ-level safety semantics and is compatible with deep RL. We demonstrate the approach on a double-integrator system by formally certifying, via satisfiability modulo theories (SMT), a value function learned through reinforcement learning to induce provably correct inner and outer backward-reachable set enclosures over a compact region of interest.
\end{abstract}


\section{Introduction}

\todo{Global remark: ensure the consistency $W_\varepsilon$ versus $V_\varepsilon$.}

Safety verification for autonomous systems hinges on computing (or certifying) the set of initial states from which a safe/unsafe region can (or cannot) be reached within a time horizon. Hamilton Jacobi (HJ) reachability encodes this guarantee via the sign of a value function that solves an HJ partial differential equation (PDE) or variational inequality (VI) in the viscosity sense \cite{mitchell2005time,lygeros2004reachability,bansal2017hj}. Classical grid based solvers suffer from the curse of dimensionality \cite{darbon2016algorithms}, and even high-order schemes remain costly at scale. Approximation strategies include decomposition \cite{chen2018decomposition}, algorithms mitigating high dimensional complexity \cite{darbon2016algorithms}, and operator-theoretic approaches (Hopf/Koopman) \cite{umathe2022reachability}.

Learning-based methods pursue scalability along two directions. (i) Self supervised reachability learning: DeepReach trains neural networks to satisfy the terminal/reach cost HJ residual, but it requires an explicit dynamics model \cite{bansal2021deepreach}. Certified Approximate Reachability (CARe) then converts bounded HJ losses into $\varepsilon$ accurate inner/outer enclosures using $\delta$ complete Satisfiability Modulo Theory (SMT) and Counterexample guided Inductive Synthesis (CEGIS) \cite{solanki2025certified}. (ii) Reinforcement Learning (RL): RL offers the possibility of learning from experience without necessarily requiring a model of the dynamics. However, standard terminal/reach cost formulations are not naturally compatible with small step discounted Bellman updates and typically do not recover exact strict reach/avoid sets. Some approaches have attempted to bridge this gap, but their presented formalism may alter safety semantics \cite{fisac2019bridging,akametalu2023minimum,wiltzer2022distributional,ganai2023iterative}.

A recently proposed discounted travel cost formulation resolves this incompatibility \cite{SolankiHJRL2025}: with cost zero off-target and strictly negative on target, the discounted Bellman iteration is provably equivalent to a forward HJB with damping, converging to its viscosity solution and recovering strict reach/avoid sets by sign, while remaining amenable to RL training. Because the learned object in this RL setting is the travel cost value (rather than the terminal/reach cost value used in DeepReach \cite{bansal2021deepreach}), CARe cannot be applied verbatim \cite{solanki2025certified}. This paper fills that gap by certifying RL learned HJ values based on the formulation in \cite{SolankiHJRL2025}.

The paper’s contributions are threefold: (i) a characterization of the discounted HJB equations corresponding to an additively shifted value function (\Cref{sec:hjb_add_eps}), (ii) certified value-error bounds derived from Bellman and PDE residuals (\Cref{sec:rl_counterpart_eps}), and (iii) SMT-based certification pipelines enabling provable reachability enclosures (\Cref{sec:smt_cert}).

\ifwithproofs
The structure of the article is as follows. \Cref{sec:problem_setup} formulates the time-invariant control problem studied in this paper and states the standing assumptions used throughout. \Cref{sec:hjb_add_eps} introduces the additive-offset identity, showing how a uniform value-function error induces a constant offset in the associated forward HJB equation. \Cref{sec:rl_counterpart_eps} then considers the discrete-time Bellman counterpart of the forward HJB : it defines the one step Bellman operator, leverages contraction and uniqueness properties to convert a certified Bellman residual bound into a uniform value function error bound, and connects this bound to reachability enclosures. \Cref{sec:smt_cert} presents two SMT based certification routes for establishing $\varepsilon_{\mathrm{val}}$ over a compact region of interest. \Cref{sec:experiment} applies the proposed certification framework to a canonical benchmark problem the double integrator. Finally, \Cref{sec:conclusion} summarizes the contributions of the paper and outlines directions for future research.
\fi

\section{Problem Setup} 
\label{sec:problem_setup}

This section formulates the time invariant optimal control problem studied in this paper. We define the system dynamics, admissible controls, running cost, and associated Hamiltonian, and state the standing regularity assumptions used throughout. We then introduce the discounted travel-time value function, its Hamilton Jacobi Bellman (HJB) characterization, and an equivalent forward/initial time formulation used in subsequent sections.

\textbf{System:}
We consider the deterministic, time-invariant control system
\begin{equation}
\label{eq:dynamics_ti}
\begin{aligned}
\dot x(s) &= f\!\big(x(s),\,u(s)\big),\\
x(t) &= x\in\R^n,\qquad s\in[t,T],
\end{aligned}
\end{equation}
where $x(\cdot)\in\R^n$ and $u(\cdot)$ is a measurable signal with values in a compact set $U\subset\R^m$.
Admissible controls on $[t,T]$ are
\begin{equation}
\mathcal M(t):=\big\{\,u:[t,T]\!\to\!U\ \text{measurable}\,\big\}.
\end{equation}
\textbf{Notation and Hamiltonians:}
For $V:[0,T]\times\R^n\to\R$, write $\Dt V:= \frac{\partial V}{\partial t}$ and $\Dx V:=\frac{\partial V}{\partial x}$.
We allow a \emph{time-dependent} running cost $h:[0,T]\times\R^n\times U\to\R$.
The (time-indexed) Hamiltonian is
\begin{equation}
\label{eq:Hamiltonian_t}
H(t,x,p)\;:=\;\inf_{u\in U}\,\Big\{\,h(t,x,u)+p\!\cdot\! f(x,u)\,\Big\},
\end{equation}
where $p:= \nabla_{x} V(t, x)$.

\begin{assumption}[Compact input set]\label{ass:A1}
The control set \(U\subset\R^m\) is nonempty and compact.
\end{assumption}

\begin{assumption}[Lipschitz dynamics in state]\label{ass:A2}
There exists \(L_f>0\) such that
\(\|f(x_1,u)-f(x_2,u)\|\le L_f\|x_1-x_2\|\) for all \(x_1,x_2\in\R^n\) and \(u\in U\).
\end{assumption}

\begin{assumption}[Uniform growth bound]\label{ass:A3}
There exists \(M_f>0\) such that \(\|f(x,u)\|\le M_f\) for all \((x,u)\in\R^n\times U\).
\end{assumption}

\begin{assumption}[Continuity in the control]\label{ass:A4}
For each fixed \(x\in\R^n\), the mapping \(u\mapsto f(x,u)\) is continuous on \(U\).
\end{assumption}

\begin{assumption}[Time dependent running cost regularity]\label{ass:A5}
The running cost \(h:[0,T]\times\R^n\times U\to\R\) is measurable and satisfies:
\begin{enumerate}\setlength{\itemsep}{2pt}
\item[(i)] (\emph{Lipschitz in \(x\), uniform in \(t,u\)}) There exists \(L_h>0\) such that
\(|h(t,x_1,u)-h(t,x_2,u)|\le L_h\|x_1-x_2\|\) for all \(t\in[0,T]\), \(x_1,x_2\in\R^n\), \(u\in U\).
\item[(ii)] (\emph{Uniform bound}) There exists \(M_h>0\) such that \(|h(t,x,u)|\le M_h\) for all \((t,x,u)\in[0,T]\times\R^n\times U\).
\item[(iii)] (\emph{Continuity in \(u\), uniform in \(t\)}) For each \((t,x)\), the map \(u\mapsto h(t,x,u)\) is continuous on \(U\).
\item[(iv)] (\emph{Time regularity, uniform on compacts}) For every compact \(K\subset\R^n\),
\[
\lim_{\delta\to 0^+}\ \sup_{\substack{t\in[0,T-\delta]\\ x\in K,\ u\in U}}
\big|h(t+\delta,x,u)-h(t,x,u)\big| \;=\;0.
\]
Equivalently, \(t\mapsto h(t,x,u)\) is continuous on \([0,T]\) uniformly for \((x,u)\in K\times U\).
\end{enumerate}
\end{assumption}


\subsection{Discounted Travel Value}
\label{subsec:discounted_value_TI}

\textbf{Discount weight:} 
Fix a discount parameter $\lambda \in \mathbb{R}$ and define the \emph{discount weight}
\begin{equation}
\label{eq:omega_def}
\omega_t(s):=e^{\lambda(t-s)}\in(0,\infty),\qquad s\in[t,T].
\end{equation}

\begin{remark}[Discount and contraction]
For the reinforcement-learning (RL) interpretation, we assume throughout that $\lambda>0$.
Then for all $s\in[t,T]$,
\[
\omega_t(s)=e^{\lambda(t-s)}=e^{-\lambda(s-t)}\in\big[e^{-\lambda(T-t)},\,1\big]\subset(0,1].
\]
In particular, $\omega_t(t)=1$ and $\omega_t(s)<1$ for all $s>t$.
Moreover, for any discrete step size $\sigma>0$, the continuation factor
$\gamma:=e^{-\lambda\sigma}\in(0,1)$, which yields a strict contraction of the
corresponding discounted Bellman operator in the sup-norm \cite{SolankiHJRL2025}.
\end{remark}

\textbf{Discounted travel value:}
We define the value directly in discounted form:
\begin{equation}
\label{eq:discounted_value_TI}
\begin{aligned}
V_\lambda(t,x)
&:= \inf_{u(\cdot)\in\mathcal M(t)}
\int_{t}^{T} \omega_t(s)\,
h\!\big(s,\,x^{u}_{t,x}(s),\,u(s)\big)\,ds,
\end{aligned}
\end{equation}
with terminal condition \(V_\lambda(T,x)=0\), where $x^{u}_{t,x}(\cdot)$ solves \cref{eq:dynamics_ti} under $u(\cdot)$.

\textbf{Discounted HJB:}
Under \Cref{ass:A1,ass:A2,ass:A3,ass:A4,ass:A5}, \(V_\lambda\) is the unique bounded viscosity solution of
\begin{equation}
\label{eq:HJB_discounted_TI}
\begin{aligned}
\Dt V_\lambda(t,x)\;+\;H\!\big(t,x,\Dx V_\lambda(t,x)\big)\;-\;\lambda\,V_\lambda(t,x) &= 0,\\
V_\lambda(T,x) &= 0.
\end{aligned}
\end{equation}
The proof of this can be found in the principle result of \cite{SolankiHJRL2025}. 


\begin{assumption}[Off-target zero]\label{ass:S0}
For a given open target \(\mathcal T\subset\R^n\),
\[ h(t,x,u)=0 \quad \text{for all } (t,x,u)\in[0,T]\times(\R^n\!\setminus\!\mathcal T)\times U. \]
\end{assumption}

\begin{assumption}[On-target negativity]\label{ass:S2}
For every \(t\in[0,T]\) and \(x\in\mathcal T\), we have
\[\inf_{u\in U} h(t,x,u)\;<\;0.\]
\end{assumption}

The strictly positive weights \(\omega_t(\cdot)\) preserve the sign logic. Therefore, as shown in \cite{SolankiHJRL2025}, under \Cref{ass:S0,ass:S2}, the sign of \(V_\lambda\) recovers the strict backward-reachability set.

\subsection{Initial/Forward Time Formulation}
\label{subsec:time_reversed_H}

Define the forward variable \(\tau:=T-t\) and \(W_\lambda(\tau,x):=V_\lambda(T-\tau,x)\).
Then \(\partial_t V_\lambda(t,x)=-\partial_\tau W_\lambda(\tau,x)\), and \cref{eq:HJB_discounted_TI} becomes
\begin{equation}
\label{eq:HJB_forward_tau}
\begin{aligned}
\partial_\tau W_\lambda(\tau,x)\;&=\; \tilde H\!\big(\tau,x,\nabla_x W_\lambda(\tau,x)\big)\;-\;\lambda\,W_\lambda(\tau,x),\qquad \\ 
W_\lambda(0,x)&=0,
\end{aligned}
\end{equation}
where the \emph{time-reversed Hamiltonian} is
\begin{equation}
\begin{aligned}
\label{eq:Hamiltonian_time_reversed}
\tilde H(\tau,x,p)\;&:=\;H\!\big(T-\tau,\,x,\,p\big)
\\ \;&=\;\inf_{u\in U}\,\Big\{\,h(T-\tau,x,u)+p\!\cdot\! f(x,u)\,\Big\}.
\end{aligned}
\end{equation}
Equations \cref{eq:HJB_discounted_TI} and \cref{eq:HJB_forward_tau} are equivalent under the change of variables \(t\leftrightarrow \tau\); \cref{eq:HJB_forward_tau} is convenient for forward in \(\tau\) dynamic programming and connects directly to discounted value iteration in RL \cite{SolankiHJRL2025}.

\section{Discounted HJB for an Additively Shifted Value}
\label{sec:hjb_add_eps}

This section establishes a central conceptual insight of the paper: a uniform value-function error is exactly equivalent to a constant offset-term perturbation in the discounted HJB equation.

\ifwithproofs
More specifically, we study an \emph{additively shifted} discounted value $V_\varepsilon:=V_\lambda+\varepsilon$ (\Cref{eq:Veps_def_add}). We establish the Dynamic Programming Principle (DPP) associated with $V_\varepsilon$ (\Cref{thm:DPP_add}). The DPP, together with \Cref{lem:local_exist_shift,lem:local_univ_shift}, then allows us to establish the viscosity sub- and supersolution inequalities for the PDE (\cref{thm:HJB_shifted_visc}) that admits $V_\varepsilon$ as its associated viscosity solution. This PDE is \cref{eq:HJB_shifted}, which is \cref{eq:HJB_discounted_TI} with an additional constant offset term $\lambda \varepsilon$. We then convert the PDE to an initial time formulation. This yields, what we have called, the forward form in which the same $\lambda\varepsilon$ offset persists (\Cref{eq:HJB_ttg}). 
\fi

\subsection{Value Function with Additive Offset}

We work under the assumptions in \Cref{sec:problem_setup}. Moreover, recall the definition of the discount weight from \cref{eq:omega_def}.

For $(t,x)\in[0,T]\times\R^n$, define the value function with additive offset as
\begin{equation}
\label{eq:Veps_def_add}
V_\varepsilon(t,x)
:= \inf_{u(\cdot)\in\mathcal M(t)}
\int_{t}^{T}\omega_t(s)\,h\!\big(s,x^{u}_{t,x}(s),u(s)\big)\,ds\;+\;\varepsilon
\end{equation}
where the trajectories follow the system dynamics shown in \cref{eq:dynamics_ti} and $\epsilon \in \mathbb{R}$ is a constant scalar. Note that
\(
V_\varepsilon(T,x)=\varepsilon
\)
and, trivially,
\(
V_\varepsilon(t,x)=V_{\lambda}(t,x)+\varepsilon
\)

\subsection{Dynamic Programming Principle Associated With Value Function with Additive Offset}
\begin{theorem}[Dynamic Programming Principle for the Additively-shifted Value $V_\varepsilon$]
\label{thm:DPP_add}
Fix $(t,x)\in[0,T)\times\R^n$ and $\sigma\in(0,T-t]$. For the additively shifted value
\[
V_\varepsilon(t,x)
= \inf_{u(\cdot)\in\mathcal M(t)} \int_{t}^{T}\!\omega_t(s)\,h\!\big(s,x^{u}_{t,x}(s),u(s)\big)\,ds\;+\;\varepsilon,
\]
the following DPP holds:
\begin{equation}
\label{eq:DPP_add}
\begin{aligned}
V_\varepsilon(t,x)
=\inf_{u(\cdot)\in\mathcal M(t)} \Big\{&
\int_{t}^{t+\sigma}\!\omega_t(s)\,h\!\big(s,x^{u}_{t,x}(s),u(s)\big)\,ds\\
&\quad +\ \omega_t(t+\sigma)\,V_\varepsilon\!\big(t+\sigma, x^{u}_{t,x}(t+\sigma)\big)\\
&\quad +\ \big(1-\omega_t(t+\sigma)\big)\,\varepsilon \Big\}.
\end{aligned}
\end{equation}
\end{theorem}

\ifwithproofs
\begin{proof}
Let $V$ denote the \emph{unshifted} discounted value,
\begin{equation}\label{eq:unshifted_value}
V_{\lambda}(t,x)
:= \inf_{u(\cdot)\in\mathcal M(t)}
\int_{t}^{T}\!\omega_t(s)\,h\!\big(s,x^{u}_{t,x}(s),u(s)\big)\,ds,
\end{equation}
so that by definition
\begin{equation}\label{eq:V_eps_as_shift}
V_\varepsilon(t,x) = V_{\lambda}(t,x) + \varepsilon,
\qquad (t,x)\in[0,T]\times\R^n.
\end{equation}

Assume that $V$ satisfies the standard discounted DPP (Lemma 2 of \cite{SolankiHJRL2025}): for any $(t,x)\in[0,T)\times\R^n$ and any $\sigma\in(0,T-t]$,
\begin{multline}\label{eq:DPP_unshifted}
V_{\lambda}(t,x)
=
\inf_{u(\cdot)\in\mathcal M(t)}
\Bigg\{
\int_t^{t+\sigma} \omega_t(s)\,
h\!\big(s,x^{u}_{t,x}(s),u(s)\big)\,ds
+ \\
\omega_t(t+\sigma)\,
V_{\lambda}\big(t+\sigma,x^{u}_{t,x}(t+\sigma)\big)
\Bigg\}.
\end{multline}

We now derive the DPP for $V_\varepsilon$ from \cref{eq:DPP_unshifted} by
simple algebra.

By \cref{eq:V_eps_as_shift} and \cref{eq:DPP_unshifted},
\begin{multline}\label{eq:Veps_from_V_step1}
    V_\varepsilon(t,x)
= V_{\lambda}(t,x) + \varepsilon \\
=
\inf_{u(\cdot)\in\mathcal M(t)}
\Bigg\{
\int_t^{t+\sigma} \omega_t(s)\,
h\!\big(s,x^{u}_{t,x}(s),u(s)\big)\,ds
+ \\ 
\omega_t(t+\sigma)\,
V\big(t+\sigma,x^{u}_{t,x}(t+\sigma)\big)
\Bigg\}
+\varepsilon.
\end{multline}
Since $\varepsilon$ does not depend on $u$, we may move it inside the infimum:
for any functional $F$ and constant $c$,
$\inf_u\{F(u)\} + c = \inf_u\{F(u)+c\}$.
Applying this to \cref{eq:Veps_from_V_step1} gives
\begin{multline}\label{eq:Veps_from_V_step2}
    V_\varepsilon(t,x)
=
\inf_{u(\cdot)\in\mathcal M(t)}
\Bigg\{
\int_t^{t+\sigma} \omega_t(s)\,
h\!\big(s,x^{u}_{t,x}(s),u(s)\big)\,ds
+ \\
\omega_t(t+\sigma)\,
V\big(t+\sigma,x^{u}_{t,x}(t+\sigma)\big)
+\varepsilon
\Bigg\}.
\end{multline}

Next, use the relation $V_{\lambda} = V_\varepsilon - \varepsilon$ from
\cref{eq:V_eps_as_shift} at time $t+\sigma$:
\[
V_{\lambda}\big(t+\sigma,x^{u}_{t,x}(t+\sigma)\big)
=
V_\varepsilon\big(t+\sigma,x^{u}_{t,x}(t+\sigma)\big) - \varepsilon.
\]
Substituting this into \cref{eq:Veps_from_V_step2} yields
\begin{multline}\label{eq:Veps_from_V_step3}
    V_\varepsilon(t,x)
=
\inf_{u(\cdot)\in\mathcal M(t)}
\Bigg\{
\int_t^{t+\sigma} \omega_t(s)\,
h\!\big(s,x^{u}_{t,x}(s),u(s)\big)\,ds \\
\hspace{4.5em}+
\omega_t(t+\sigma)\,
\Big(V_\varepsilon\big(t+\sigma,x^{u}_{t,x}(t+\sigma)\big)-\varepsilon\Big)
+\varepsilon
\Bigg\}.
\end{multline}

Now expand the last two terms inside the braces:
\begin{align*}
&\omega_t(t+\sigma)\,\Big(V_\varepsilon(\cdot)-\varepsilon\Big)
+\varepsilon\\
&=
\omega_t(t+\sigma)\,V_\varepsilon\big(t+\sigma,x^{u}_{t,x}(t+\sigma)\big)
-\omega_t(t+\sigma)\,\varepsilon
+\varepsilon\\
&=
\omega_t(t+\sigma)\,V_\varepsilon\big(t+\sigma,x^{u}_{t,x}(t+\sigma)\big)
+
\big(1-\omega_t(t+\sigma)\big)\,\varepsilon.
\end{align*}
Substituting this back into \cref{eq:Veps_from_V_step3}, we precisely obtain \cref{eq:DPP_add}. This completes the proof.
\end{proof}
\fi

\subsection{HJB PDE Associated with $V_\varepsilon$}

This subsection first presents \Cref{lem:local_univ_shift,lem:local_exist_shift,lem:local_univ_shift} which are used to establish the local implications between pointwise violations of the HJB inequality and violations of the short-horizon DPP. Thereafter, \Cref{thm:HJB_shifted_visc} establishes an exact correspondence between uniform additive errors in the value function and constant-offset HJB inequalities.

\begin{lemma}\label{lem:local_exist_shift}
 Let a smooth test function $\phi\in C^1([0,T]\times\R^n)$ and suppose
at $(t_0,x_0)$
\begin{multline}\label{eq:LES_cond_fixed}
\phi_t(t_0,x_0)+H\!\big(t_0,x_0,D\phi(t_0,x_0)\big)
\\ -\lambda\,\phi(t_0,x_0)+\lambda\,\varepsilon \;\le\; -\theta,
\quad \theta>0.
\end{multline}
Then there exist $u^\ast\in U$ and $\delta_0>0$ such that, for the trajectory
$x(\cdot)$ solving $\dot x=f(x,u^\ast)$ with $x(t_0)=x_0$, and all
$\delta\in(0,\delta_0]$,
\begin{multline}\label{eq:LES_goal_fixed}
     e^{-\lambda\delta}\,\phi(t_0+\delta,x(\delta))-\phi(t_0,x_0) 
+ \\ \!\int_{0}^{\delta}\! e^{-\lambda r}\,h\!\big(t_0+r,x(r),u^\ast\big)\,dr \\
+\big(1-e^{-\lambda\delta}\big)\,\varepsilon
\;\le\; -\frac{\theta}{2}\!\int_{0}^{\delta} e^{-\lambda r}\,dr .
\end{multline}
\end{lemma}

\ifwithproofs
\begin{proof}
Assume \eqref{eq:LES_cond_fixed} holds at $(t_0,x_0)$ for some $\theta>0$.
Fix $\eta\in(0,\theta/4)$. By definition of the Hamiltonian as an infimum,
there exists $u^\ast\in U$ such that
\begin{multline}\label{eq:u_star_selection_rig}
    h(t_0,x_0,u^\ast)+\nabla_x\phi(t_0,x_0)\cdot f(x_0,u^\ast) \\
\le H\!\big(t_0,x_0,\nabla_x\phi(t_0,x_0)\big)+\eta
\end{multline}

Let $x(\cdot)$ solve $\dot x(s)=f(x(s),u^\ast)$ with $x(t_0)=x_0$, and define
the shifted-time trajectory $y(r):=x(t_0+r)$ for $r\in[0,\delta]$.
Define
\begin{multline}\label{eq:Fdef_exist_rig}
    F(\delta)
:= e^{-\lambda\delta}\phi(t_0+\delta,y(\delta))-\phi(t_0,x_0) \\
+\int_{0}^{\delta} e^{-\lambda r}\,h(t_0+r,y(r),u^\ast)\,dr \\
 +\big(1-e^{-\lambda\delta}\big)\varepsilon
\end{multline}

\noindent\emph{Step 1: Exact integral representation.}
Let $\Psi(r):=e^{-\lambda r}\phi(t_0+r,y(r))$. Since $\phi\in C^1$ and $y$ is
absolutely continuous, $\Psi$ is absolutely continuous and for a.e.\ $r$,
\begin{multline*}
    \Psi'(r)=e^{-\lambda r}\Big(\phi_t(t_0+r,y(r)) \\
+\nabla_x\phi(t_0+r,y(r))\cdot f(y(r),u^\ast)
-\lambda\phi(t_0+r,y(r))\Big)
\end{multline*}
Integrating from $0$ to $\delta$ and using
$(1-e^{-\lambda\delta})\varepsilon=\int_0^\delta \lambda e^{-\lambda r}\varepsilon\,dr$
yields the exact identity
\begin{equation}\label{eq:Fexact_exist_rig}
F(\delta)=\int_{0}^{\delta} e^{-\lambda r}\,G(t_0+r,y(r),u^\ast)\,dr,
\end{equation}
where
\begin{multline}\label{eq:Gdef_exist_rig}
    G(t,x,u):=\phi_t(t,x)+\nabla_x\phi(t,x)\cdot f(x,u) \\ -\lambda\phi(t,x)+h(t,x,u)+\lambda\varepsilon
\end{multline}

\noindent\emph{Step 2: Strict negativity at $(t_0,x_0)$ for the chosen $u^\ast$.}
By \eqref{eq:LES_cond_fixed} and \eqref{eq:u_star_selection_rig},
\begin{multline*}
    G(t_0,x_0,u^\ast)
= \phi_t(t_0,x_0)-\lambda\phi(t_0,x_0)+\lambda\varepsilon \\
+\nabla_x\phi(t_0,x_0)\cdot f(x_0,u^\ast)+h(t_0,x_0,u^\ast) \\
\le \phi_t(t_0,x_0)-\lambda\phi(t_0,x_0)+\lambda\varepsilon \\
+H\!\big(t_0,x_0,\nabla_x\phi(t_0,x_0)\big)+\eta 
\le -\theta+\eta \ \le\ -\tfrac{3}{4}\theta
\end{multline*}
Hence $G(t_0,x_0,u^\ast)\le -\tfrac{3}{4}\theta$.

\noindent\emph{Step 3: Persistence on a neighborhood.}
Since $G$ is continuous in $(t,x)$ for fixed $u^\ast$ (by $\phi\in C^1$ and
continuity of $f,h$ in $t,x$), there exist $\delta_1>0$ and $\rho>0$ such that
\begin{multline}\label{eq:Gneg_neighborhood}
    G(t,x,u^\ast)\le -\tfrac{1}{2}\theta
 \\ \text{for all}\quad t\in[t_0,t_0+\delta_1],\ \|x-x_0\|\le \rho
\end{multline}

\noindent\emph{Step 4: Keep the trajectory inside the neighborhood.}
Using the growth bound $\|f(x,u)\|\le M_f$, we have for $r\in[0,\delta]$,
\[
\|y(r)-x_0\|
\le \int_0^r \|f(y(s),u^\ast)\|ds
\le M_f r.
\]
Choose
\[
\delta_0:=\min\{\delta_1,\ \rho/M_f,\ T-t_0\}.
\]
Then for any $\delta\in(0,\delta_0]$ and all $r\in[0,\delta]$,
$t_0+r\in[t_0,t_0+\delta_1]$ and $\|y(r)-x_0\|\le \rho$, so
\eqref{eq:Gneg_neighborhood} gives
\[
G(t_0+r,y(r),u^\ast)\le -\tfrac{1}{2}\theta,\qquad \forall r\in[0,\delta].
\]

\noindent\emph{Step 5: Integrate to conclude.}
Using \eqref{eq:Fexact_exist_rig},
\[
F(\delta)=\int_{0}^{\delta} e^{-\lambda r}\,G(t_0+r,y(r),u^\ast)\,dr
\le -\tfrac{1}{2}\theta\int_{0}^{\delta} e^{-\lambda r}\,dr
\]
which is exactly \eqref{eq:LES_goal_fixed}.
\end{proof}
\fi


\begin{lemma}\label{lem:local_univ_shift} 
Let $\phi\in C^1([0,T]\times\R^n)$ and suppose
at $(t_0,x_0)$
\begin{multline}\label{eq:local_univ_shift_cond}
\phi_t(t_0,x_0)+H\!\big(t_0,x_0,D\phi(t_0,x_0)\big) \\
-\lambda\,\phi(t_0,x_0)+\lambda\, \varepsilon \;\ge\; \theta,
\quad \theta>0.
\end{multline}
Then, for every measurable control $u(\cdot)$ on $[t_0,t_0+\delta]$,
there exists $\delta_0>0$ such that, for the trajectory $x(\cdot)$ solving
$\dot x=f(x,u(\cdot))$ with $x(t_0)=x_0$, and all $\delta\in(0,\delta_0]$,
\begin{multline}\label{eq:local_univ_shift_goal}
     e^{-\lambda\delta}\,\phi(t_0+\delta,x(\delta))-\phi(t_0,x_0) \\
+\!\int_{0}^{\delta}\! e^{-\lambda r}\,h\!\big(t_0+r,x(r),u(r)\big)\,dr \\
\qquad\qquad\qquad
+\big(1-e^{-\lambda\delta}\big)\,\varepsilon
\;\ge\; \frac{\theta}{2}\!\int_{0}^{\delta} e^{-\lambda r}\,dr .
\end{multline}
\end{lemma}

\ifwithproofs
\begin{proof}
Fix an arbitrary measurable control $u(\cdot)$ on $[t_0,t_0+\delta]$ and let
$x(\cdot)$ be the corresponding trajectory solving
\[
\dot x(s)=f(x(s),u(s)),\qquad x(t_0)=x_0.
\]
For notational convenience, define the shifted-time trajectory
\[
y(r):=x(t_0+r),\qquad r\in[0,\delta],
\]
so that $y(0)=x_0$ and $\dot y(r)=f(y(r),u(t_0+r))$ for a.e.\ $r$.

Define the short-horizon functional
\begin{multline}\label{eq:Fdef_univ_rig}
    F(\delta)
:= e^{-\lambda\delta}\phi(t_0+\delta,y(\delta))-\phi(t_0,x_0) \\
+\int_{0}^{\delta} e^{-\lambda r}\,h(t_0+r,y(r),u(t_0+r))\,dr \\
 +\big(1-e^{-\lambda\delta}\big)\varepsilon
\end{multline}

\noindent\emph{Step 1: Exact integral representation.}
Let
\[
\Psi(r):=e^{-\lambda r}\phi(t_0+r,y(r)),\qquad r\in[0,\delta].
\]
Since $\phi\in C^1$ and $y(\cdot)$ is absolutely continuous, $\Psi$ is
absolutely continuous and for a.e.\ $r\in[0,\delta]$,
\begin{multline*}
    \Psi'(r)
=-\lambda e^{-\lambda r}\phi(t_0+r,y(r)) \\
+e^{-\lambda r}\Big(\phi_t(t_0+r,y(r))+\nabla_x\phi(t_0+r,y(r))\cdot \dot y(r)\Big)\\
=e^{-\lambda r}\Big(\phi_t(t_0+r,y(r))
+\nabla_x\phi(t_0+r,y(r))\cdot f(y(r),u(t_0+r)) \\
-\lambda\phi(t_0+r,y(r))\Big)
\end{multline*}
Integrating from $0$ to $\delta$ gives
\begin{multline}\label{eq:chainint_univ_rig}
    e^{-\lambda\delta}\phi(t_0+\delta,y(\delta))-\phi(t_0,x_0)
= \\ \int_{0}^{\delta} e^{-\lambda r}
\Big(\phi_t+\nabla_x\phi\cdot f-\lambda\phi\Big)(t_0+r,y(r),u(t_0+r))\,dr
\end{multline}
Moreover,
\begin{equation}\label{eq:epsint_univ_rig}
(1-e^{-\lambda\delta})\varepsilon=\int_{0}^{\delta}\lambda e^{-\lambda r}\varepsilon\,dr.
\end{equation}
Substituting \eqref{eq:chainint_univ_rig} and \eqref{eq:epsint_univ_rig} into
\eqref{eq:Fdef_univ_rig} yields the exact identity
\begin{equation}\label{eq:Fexact_univ_rig}
F(\delta)=\int_{0}^{\delta} e^{-\lambda r}\,G(t_0+r,y(r),u(t_0+r))\,dr,
\end{equation}
where
\begin{multline}\label{eq:Gdef_univ_rig}
    G(t,x,u):= \phi_t(t,x)+\nabla_x\phi(t,x)\cdot f(x,u)- \\ \lambda\phi(t,x)+h(t,x,u)+\lambda\varepsilon
\end{multline}

\noindent\emph{Step 2: Pointwise lower bound at $(t_0,x_0)$.}
By definition of the Hamiltonian as an infimum,
\[
h(t_0,x_0,u)+\nabla_x\phi(t_0,x_0)\cdot f(x_0,u)\ \ge\
H\!\big(t_0,x_0,\nabla_x\phi(t_0,x_0)\big),\qquad \forall u\in U.
\]
Adding $\phi_t(t_0,x_0)-\lambda\phi(t_0,x_0)+\lambda\varepsilon$ to both sides and
using \eqref{eq:local_univ_shift_cond} yields
\begin{equation}\label{eq:Gtheta_univ_rig}
G(t_0,x_0,u)\ \ge\ \theta,\qquad \forall u\in U.
\end{equation}
Define
\[
m(t,x):=\min_{u\in U}G(t,x,u).
\]
Then \eqref{eq:Gtheta_univ_rig} implies $m(t_0,x_0)\ge \theta$.

\noindent\emph{Step 3: Uniform persistence on a neighborhood.}
Since $\phi\in C^1$ and $(t,x,u)\mapsto (f(x,u),h(t,x,u))$ is continuous and $U$
is compact, the map $G$ is continuous and thus $m(t,x)=\min_{u\in U}G(t,x,u)$ is
continuous in $(t,x)$. Hence there exist $\delta_1>0$ and $\rho>0$ such that
\begin{equation}\label{eq:mhalf_univ_rig}
m(t,x)\ge \theta/2
\quad\text{for all}\quad t\in[t_0,t_0+\delta_1],\ \|x-x_0\|\le \rho.
\end{equation}
Equivalently, $G(t,x,u)\ge \theta/2$ for all $u\in U$ on this neighborhood.

\noindent\emph{Step 4: Keep the trajectory inside the neighborhood.}
By the uniform growth bound $\|f(x,u)\|\le M_f$,
\begin{multline*}
    \|y(r)-x_0\|\le \int_0^r \|\dot y(s)\|\,ds \\
=\int_0^r \|f(y(s),u(t_0+s))\|\,ds
\le \int_0^r M_f\,ds
= M_f r
\end{multline*}
Choose
\[
\delta_0:=\min\{\delta_1,\ \rho/M_f,\ T-t_0\}.
\]
Then for any $\delta\in(0,\delta_0]$ and all $r\in[0,\delta]$ we have
$t_0+r\in[t_0,t_0+\delta_1]$ and $\|y(r)-x_0\|\le \rho$, so by
\eqref{eq:mhalf_univ_rig},
\[
G(t_0+r,y(r),u(t_0+r))\ge \theta/2,\qquad \forall r\in[0,\delta].
\]

\noindent\emph{Step 5: Integrate to conclude.}
Using \eqref{eq:Fexact_univ_rig},
\begin{multline*}
    F(\delta)=\int_{0}^{\delta} e^{-\lambda r}\,G(t_0+r,y(r),u(t_0+r))\,dr \\
\ge \int_{0}^{\delta} e^{-\lambda r}\,\frac{\theta}{2}\,dr
=\frac{\theta}{2}\int_{0}^{\delta} e^{-\lambda r}\,dr
\end{multline*}
This is exactly \eqref{eq:local_univ_shift_goal}.
\end{proof}
\fi

\begin{theorem}[Viscosity HJB for $V_\varepsilon$]\label{thm:HJB_shifted_visc}
Under the assumptions of \Cref{sec:problem_setup}, the function
\begin{equation}
\begin{aligned}
V_\varepsilon(t,x)
=\inf_{u(\cdot)\in\mathcal M(t)}\int_t^T \omega_t(s)\,h\!\big(s,x^u_{t,x}(s),u(s)\big)\,ds
\;+\;\varepsilon ,
\end{aligned}
\end{equation}
is the unique bounded, uniformly continuous viscosity solution of
\begin{equation}
\label{eq:HJB_shifted}
\begin{aligned}
\Dt V_\varepsilon(t,x)
\;+\; H\!\big(t,x,\Dx V_\varepsilon(t,x)\big)
\;-\; \lambda\,V_\varepsilon(t,x)
\;+\; \lambda\,\varepsilon &= 0,
\\
V_\varepsilon(T,x)&=\varepsilon.
\end{aligned}
\end{equation}
\end{theorem}

\ifwithproofs
\begin{proof}
We verify the viscosity sub- and super-solution inequalities; the terminal condition is immediate from the definition. Uniqueness follows from comparison for proper continuous Hamiltonians.

\textbf{(i) Subsolution.}
Let $\phi\in C^1$ and suppose $V_\varepsilon-\phi$ attains a local maximum $0$ at $(t_0,x_0)$.
We must show that
\begin{equation}
\label{eq:sub_goal}
\phi_t(t_0,x_0)+H\!\big(t_0,x_0,D\phi(t_0,x_0)\big)
-\lambda\,\phi(t_0,x_0)+\lambda\,\varepsilon\ \ge\ 0.
\end{equation}
Assume by contradiction that \cref{eq:sub_goal} is not true, then there would exist $\theta>0$ such that
\begin{equation}
\label{eq:sub_contra}
\phi_t(t_0, x_0)+H(t_0,x_0,D\phi(t_0, x_0))-\lambda\phi(t_0, x_0)+\lambda\varepsilon \ \le\ -\theta.
\end{equation}
Applying \Cref{lem:local_exist_shift}, there exists a control
$u^\ast$ and $\delta_0>0$ such that, for $\delta\in(0,\delta_0]$ and the corresponding trajectory $x(\cdot)$ satisfies,
\begin{equation}
\begin{aligned}
\label{eq:exist_ineq_use}
&e^{-\lambda\delta}\phi(t_0+\delta,x(\delta))-\phi(t_0,x_0)
\; \\
&+\;\int_{0}^{\delta}\! e^{-\lambda r}\,h(t_0+r,x(r),u^\ast)\,dr\;+\;\big(1-e^{-\lambda\delta}\big)\varepsilon \\
&\ \le\ -\frac{\theta}{2}\!\int_0^\delta e^{-\lambda r}\,dr. 
\end{aligned}    
\end{equation}

Based on the premise that $V_\varepsilon-\phi$ attains a local maximum $0$ at $(t_0,x_0)$, this means that $\phi = V_\varepsilon$ at $(t_0,x_0)$, and $\phi\ge V_\varepsilon$ for values different from $(t_0,x_0)$. This means that we have:
\begin{multline}
\label{eq:touch_above_use}
e^{-\lambda\delta}\,V_\varepsilon(t_0+\delta,x(\delta)) - V_\varepsilon(t_0,x_0)
\\ \le\ e^{-\lambda\delta}\phi(t_0+\delta,x(\delta))-\phi(t_0,x_0).
\end{multline}
Recall that the DPP with offset (short-step form at $(t_0,x_0)$) based on \cref{eq:DPP_add} means that
\begin{equation}
\label{eq:DPP_short}
\begin{aligned}
V_\varepsilon(t_0,x_0)
&=\inf_{u(\cdot)}\Big\{\int_{0}^{\delta}\! e^{-\lambda r}\,h(t_0+r,x(r),u(r))\,dr
 \\ & + e^{-\lambda\delta} V_\varepsilon(t_0+\delta,x(\delta))
+\ \big(1-e^{-\lambda\delta}\big)\varepsilon\Big\} .
\end{aligned}
\end{equation}

Adding $\int_0^\delta e^{-\lambda r} h(\cdot) d r+\left(1-e^{-\lambda \delta}\right) \varepsilon$ to both sides of \cref{eq:touch_above_use} yields
\begin{equation}
\begin{aligned}
&e^{-\lambda\delta}\,V_\varepsilon(t_0+\delta,x(\delta)) - V_\varepsilon(t_0,x_0) + \\&\int_{0}^{\delta}\! e^{-\lambda r}\,h(t_0+r,x(r),u(r))\,dr + \big(1-e^{-\lambda\delta}\big)\varepsilon \\
&\le\ e^{-\lambda\delta}\phi(t_0+\delta,x(\delta))-\phi(t_0,x_0) + \\
&\int_{0}^{\delta}\! e^{-\lambda r}\,h(t_0+r,x(r),u(r))\,dr + \big(1-e^{-\lambda\delta}\big)\varepsilon .
\end{aligned}
\end{equation}
This means that
\begin{equation}
\begin{aligned}
    &e^{-\lambda \delta} V_{\varepsilon}\left(t_0+\delta, x(\delta)\right)-V_{\varepsilon}\left(t_0, x_0\right)\\
    &+\int_0^\delta e^{-\lambda r} h(\cdot) d r+\left(1-e^{-\lambda \delta}\right) \varepsilon \leq -\frac{\theta}{2} \int_0^\delta e^{-\lambda r} d r .
\end{aligned}
\end{equation}

Taking the infimum over $u(\cdot)$, we get
\begin{equation}
\begin{aligned}
&\inf_{u(\cdot)}\Big\{\int_{0}^{\delta}\! e^{-\lambda r}\,h(t_0+r,x(r),u(r))\,dr
 \\ & + e^{-\lambda\delta} V_\varepsilon(t_0+\delta,x(\delta))
+\ \big(1-e^{-\lambda\delta}\big)\varepsilon\Big\} \\
&- V_\varepsilon(t_0,x_0) \leq -\frac{\theta}{2} \int_0^\delta e^{-\lambda r} d r .
\end{aligned}
\label{eq:inf_less_than_theta}
\end{equation}

Since, based on \cref{eq:DPP_short}, the first part of the \cref{eq:inf_less_than_theta} is $V_\varepsilon(t_0,x_0)$, we obtain the contradiction
$0 \le -\tfrac{\theta}{2}\int_0^\delta e^{-\lambda r}\,dr<0$ for small $\delta$.
Hence, \cref{eq:sub_goal} holds.

\textbf{(ii) Supersolution.}
Let $\phi\in C^1$ and suppose $V_\varepsilon-\phi$ attains a local minimum $0$ at $(t_0,x_0)$.
We must show
\begin{multline}
\label{eq:super_goal}
\phi_t(t_0,x_0)+H\!\big(t_0,x_0,D\phi(t_0,x_0)\big)
-\lambda\,\phi(t_0,x_0)+\lambda\,\varepsilon\ \\ \le\ 0.
\end{multline}
Assume by contradiction that \cref{eq:super_goal} is not true, then there would exist $\theta>0$ with 
\begin{equation}
\label{eq:super_contra}
\phi_t(t_0, x_0)+H(t_0,x_0,D\phi(t_0, x_0))-\lambda\phi(t_0, x_0)+\lambda\varepsilon \ \ge\ \theta.
\end{equation}
Applying \Cref{lem:local_univ_shift} with this $\theta$ to any measurable control $u(\cdot)$ on $[t_0,t_0+\delta]$ and the corresponding trajectory $x(\cdot)$ and for small $\delta>0$,
\begin{multline}\label{eq:univ_ineq_use}
e^{-\lambda\delta}\phi(t_0+\delta,x(\delta))-\phi(t_0,x_0)
\;+\; \\ \int_{0}^{\delta}\! e^{-\lambda r}\,h(t_0+r,x(r),u(r))\,dr
\\[-2pt]
\;+\;\big(1-e^{-\lambda\delta}\big)\varepsilon
\ \ge\ \frac{\theta}{2}\!\int_0^\delta e^{-\lambda r}\,dr .
\end{multline}

Since now $\phi\le V_\varepsilon$ near $(t_0,x_0)$ with equality at $(t_0,x_0)$,
\begin{equation}
\begin{aligned}
\label{eq:touch_below_use}
&e^{-\lambda\delta}\,V_\varepsilon(t_0+\delta,x(\delta)) - V_\varepsilon(t_0,x_0)\\
&\ \ge\ e^{-\lambda\delta}\phi(t_0+\delta,x(\delta))-\phi(t_0,x_0).
\end{aligned}
\end{equation}
Following the steps of the proof for the subsolution, we make use of \cref{eq:exist_ineq_use} and \cref{eq:touch_below_use}, and then use the DPP \cref{eq:DPP_short} with the any $u$ to reach $0 < \tfrac{\theta}{2}\int_0^\delta e^{-\lambda r}\,dr\leq 0$ for small $\delta$. This is a contradiction. Thus, \cref{eq:super_goal} holds.

\textbf{(iii) Boundary condition.}
By definition $V_\varepsilon(T,x)=\varepsilon$.
\end{proof}
\fi

\subsection{Forward/initial time formulation and corresponding HJB}
\label{subsec:ttg_statement}

Define the time-to-go variable \(\tau:=T-t\) and
\begin{equation}
\label{eq:W_def}
W_\varepsilon(\tau,x)\;:=\;V_\varepsilon(T-\tau,x),\qquad \tau\in[0,T].
\end{equation}
Equivalently, using the change of variables \(r:=s-(T-\tau)\),
\begin{equation}
\label{eq:W_cost}
W_\varepsilon(\tau,x)
\;=\; \inf_{u(\cdot)} \int_{0}^{\tau} e^{-\lambda r}\,
h\!\big(T-\tau+r,\;x^{u}(r),\;u(r)\big)\,dr \;+\; \varepsilon,
\end{equation}
with \(x^{u}(0)=x\) and \(\dot x=f(x,u)\).
Introduce the time-reversed Hamiltonian
\begin{equation}
\label{eq:tilde_H}
\tilde{H}(\tau,x,p)\;:=\;H(T-\tau,x,p).
\end{equation}
The \(W_\varepsilon\) is the viscosity solution of HJB  
\begin{equation}
\label{eq:HJB_ttg}
\begin{aligned}
\partial_\tau W_\varepsilon(\tau,x)
\;&=\;\tilde H\!\big(\tau,x,\nabla_x W_\varepsilon(\tau,x)\big)
\;-\;\lambda\,W_\varepsilon(\tau,x)\;+\;\lambda\,\varepsilon,\\
W_\varepsilon(0,x)&=\varepsilon.
\end{aligned}
\end{equation}

\section{RL counterpart and $\varepsilon$-residual implications}
\label{sec:rl_counterpart_eps}

Building on the forward HJB formulation in \cref{eq:HJB_ttg}, we now pass to the discrete-time Bellman view that enables certification. We introduce the one-step operator $T_{\sigma,\lambda}$ in \cref{eq:Bellman_T_discrete_clean}. For a discount rate $\lambda>0$ and discount factor $\gamma=e^{-\lambda\sigma}$, the operator is a strict contraction and thus admits a unique fixed-point $W_\lambda$ (\cref{eq:W_fixed_point_clean}). This is based on the Banach fixed-point theorem \cite{van2022functional}.

Consequently, any certified bound on the Bellman residual yields a uniform (sup-norm) $L_\infty$ value-function error bound between the learned approximation $\widehat W$ and the true value $W_{\lambda}$, characterized by \cref{eq:value_error_bound_clean} in \Cref{lem:contr_value_err}. We denote this bound by $\varepsilon_{\mathrm{val}}$.

Interpreting $\widehat W$ in a viscosity sense and leveraging $\varepsilon_{\mathrm{val}}$ yields the shifted forward HJB that characterize each of $\widehat W \pm \varepsilon_{\mathrm{val}}$  (\Cref{thm:HJB_shifted_visc}) in the viscosity sense. Finally, under the strict reachability sign semantics assumptions, thresholding $\widehat W$ at $\pm \varepsilon_{\mathrm{val}}$ produces calibrated inner/outer enclosures of the backward-reachable tube, which is proved in \Cref{thm:reach_bracket}.

Specifically, \Cref{subsec:discrete-time-ttg} defines $T_{\sigma,\lambda}$ and the TTG fixed-point identity; \Cref{subsec:bellman-operator-contraction} shows contraction implies residual-to-value error bounds; \Cref{thm:reach_bracket} converts that bound into shifted forward HJB envelopes and reachable-tube enclosures.

\subsection{Discrete-time Bellman Operator}
\label{subsec:discrete-time-ttg}
We work with the initial-time value function $W_\lambda(\tau,x)$, defined in \cref{eq:HJB_forward_tau} cf.\ \Cref{subsec:time_reversed_H}.

Fix a timestep $\sigma\in(0,T]$ and set
$\gamma:=e^{-\lambda\sigma}\in(0,1)$. Let
$\mathcal A_\sigma:=\{\,a:[0,\sigma]\!\to\!U\ \text{measurable}\,\}$ and, for state trajectory $\frac{dy}{dt}(r)=f(y(r),a(r))$, $y(0)=x$, define the one-step discounted cost
\begin{equation}
\label{eq:one_step_cost}
c(\tau,x,a):=\int_{0}^{\sigma} e^{-\lambda r}\,h\!\big(T-\tau+r,\,y(r),\,a(r)\big)\,dr.
\end{equation}
We define the Bellman operator on bounded $\Psi:[0,T]\times\R^n\to\R$ as
\begin{equation}
\label{eq:Bellman_T_discrete_clean}
(T_{\sigma,\lambda}\Psi)(\tau,x)
:= \inf_{a\in\mathcal A_\sigma}\Big\{\,c(\tau,x,a)+\gamma\,\Psi\big(\tau-\sigma,\,y(\sigma)\big)\Big\}.
\end{equation}
Applying the DPP yields the fixed-point identity on $[\sigma,T]\times\R^n$:
\begin{equation}
\label{eq:W_fixed_point_clean}
W_\lambda \;=\; T_{\sigma,\lambda} W_\lambda .
\end{equation}
(See Sec. 4.1 of \cite{SolankiHJRL2025}).

\subsection{Bellman Operator Value-error and HJB PDE slack Bound Guarantees}
\label{subsec:bellman-operator-contraction}

\todo{Check the correctness of the region of validity.}

\Cref{lem:contr_value_err} allows us to establish the global worst-case bound on the learned value-function error over the region of interest, directly induced by a Bellman residual and \cref{thm:slack_implies_bracket_offset_roi} enables us to use the HJB PDE slack to bound the worst-case bound on the learned value-function.

\begin{lemma}[Contraction and value error]\label{lem:contr_value_err}
If $\lambda>0$, then $T_{\sigma,\lambda}$ is a $\gamma$-contraction in $\|\cdot\|_\infty$ on $[\sigma,T]\times\R^n$:
\begin{equation}
\begin{aligned}
\label{eq:gamma-contraction}
\|T_{\sigma,\lambda}\Psi_1 - T_{\sigma,\lambda}\Psi_2\|_\infty
\;&\le\; \gamma\,\|\Psi_1-\Psi_2\|_\infty,\\
\gamma&=e^{-\lambda\sigma}\in(0,1).
\end{aligned}
\end{equation}
Consequently, if the Bellman residual of $\widehat W$ is upper bounded
\begin{equation}\label{eq:residual_bound_clean}
\big\|T_{\sigma,\lambda}\widehat W-\widehat W\big\|_\infty\le \varsigma,
\end{equation}
where $\varsigma$ is a certified residual bound (determined, for example, through an SMT solver) then, 
\begin{equation}
\label{eq:value_error_bound_clean}
\|\widehat W - W_\lambda\|_\infty
\;\le\; \frac{\varsigma}{1-\gamma}
\;=\; \frac{\varsigma}{1-e^{-\lambda\sigma}}
\;=:\; \varepsilon_{\mathrm{val}}.
\end{equation}

\end{lemma}

\begin{proof}
    Starting with $\|\widehat W - W_\lambda\|_\infty$ and adding and subtracting $T_{\sigma,\lambda} \widehat W$ yields
    \begin{equation}
        \|\widehat W - W_\lambda\|_\infty = \|\widehat W + T_{\sigma,\lambda} \widehat W - T_{\sigma,\lambda} \widehat W - W_\lambda\|_\infty .
    \end{equation}
    From the fixed-point identity, shown in \cref{eq:W_fixed_point_clean}, we can replace $W_\lambda$ with $T_{\sigma,\lambda} W_\lambda$:
    \begin{equation}
        \begin{aligned}
        \|\widehat W - W_\lambda\|_\infty &= \|\widehat W + T_{\sigma,\lambda} \widehat W - T_{\sigma,\lambda} \widehat W - T_{\sigma,\lambda} W_\lambda\|_\infty\\
        &= \|(\widehat W - T_{\sigma,\lambda} \widehat W) + (T_{\sigma,\lambda} \widehat W - T_{\sigma,\lambda} W_\lambda)\|_{\infty} .
        \end{aligned}
    \end{equation}
    Based on the triangle inequality, 
    \begin{equation}
        \begin{aligned}
        &\|(\widehat W - T_{\sigma,\lambda} \widehat W) + (T_{\sigma,\lambda} \widehat W - T_{\sigma,\lambda} W_\lambda)\|_{\infty} \\
        &\leq \|(\widehat W - T_{\sigma,\lambda} \widehat W)\|_\infty + \|(T_{\sigma,\lambda} \widehat W - T_{\sigma,\lambda} W_\lambda)\|_{\infty} .
        \end{aligned}
    \end{equation}
    Leveraging \cref{eq:gamma-contraction} and \cref{eq:residual_bound_clean}, 
    \begin{equation}
        \|\widehat W - W_\lambda\|_\infty
        \leq \varsigma + \gamma \|\widehat W - W_\lambda\|_{\infty} .
    \end{equation}
    Rearranging and simplifying yields \cref{eq:value_error_bound_clean}.
\end{proof}


\begin{theorem}[HJB-slack implies uniform enclosure of $W_\lambda$ on an ROI]
\label{thm:slack_implies_bracket_offset_roi}
Let $\lambda>0$ and $\tilde H(\tau,x,p):=H(T-\tau,x,p)$. Fix a compact ROI $X\subset\R^n$ and set
$\mathcal D:=[0,T]\times X$. Let $W_\lambda$ denote the unique bounded, uniformly continuous viscosity solution of
\begin{multline}\label{eq:forward_HJB_base_slack_offset_roi}
\partial_\tau W_\lambda(\tau,x)
=\tilde H\!\big(\tau,x,\nabla_x W_\lambda(\tau,x)\big)-\lambda\,W_\lambda(\tau,x),
\qquad \\
W_\lambda(0,x)=0.
\end{multline}
Let $\widehat W:\mathcal D\to\R$ be bounded and uniformly continuous. Assume there exist
$\varepsilon_{\mathrm{pde}}\ge 0$ and $\varepsilon_0\ge 0$ such that:

\begin{enumerate}
\item[\textup{(i)}] (\emph{$\varepsilon_{\mathrm{pde}}$-HJB slack on $(0,T]$ in the viscosity sense})
For every $\phi\in C^1([0,T]\times\R^n)$ and every $(\tau_0,x_0)\in(0,T]\times X$,
\begin{multline}\label{eq:pde_slack_sub_offset_roi}
\text{if $\widehat W-\phi$ has a local maximum at $(\tau_0,x_0)$, then}\\
\partial_\tau\phi(\tau_0,x_0)-\tilde H(\tau_0,x_0,\nabla_x\phi(\tau_0,x_0))+\lambda\,\phi(\tau_0,x_0)
\ \\ \le\ \varepsilon_{\mathrm{pde}},
\end{multline}
\begin{multline}\label{eq:pde_slack_sup_offset_roi}
\text{if $\widehat W-\phi$ has a local minimum at $(\tau_0,x_0)$, then}\\
\partial_\tau\phi(\tau_0,x_0)-\tilde H(\tau_0,x_0,\nabla_x\phi(\tau_0,x_0))+\lambda\,\phi(\tau_0,x_0)
\ \\ \ge\ -\varepsilon_{\mathrm{pde}}.
\end{multline}

\item[\textup{(ii)}] (\emph{initial mismatch bound on $X$})
\begin{equation}\label{eq:init_mismatch_offset_roi}
\sup_{x\in X}\big|\widehat W(0,x)\big|\le \varepsilon_0.
\end{equation}
\end{enumerate}

Define
\begin{equation}\label{eq:def_eps_val_offset_roi}
\varepsilon_{\mathrm{val}}
:=\max\Big\{\frac{\varepsilon_{\mathrm{pde}}}{\lambda},\ \varepsilon_0\Big\},
\end{equation}
and the envelopes on $\mathcal D$,
\begin{equation}\label{eq:env_def_offset_roi}
\underline W:=\widehat W-\varepsilon_{\mathrm{val}},
\qquad
\overline W:=\widehat W+\varepsilon_{\mathrm{val}}.
\end{equation}
Then, on $\mathcal D$,
\begin{equation}\label{eq:bracket_conclusion_offset_roi}
\underline W(\tau,x)\ \le\ W_\lambda(\tau,x)\ \le\ \overline W(\tau,x),
\qquad \forall(\tau,x)\in\mathcal D,
\end{equation}
and in particular
\begin{equation}\label{eq:supnorm_bound_from_slack_offset_roi}
\|\widehat W-W_\lambda\|_{L^\infty(\mathcal D)} \le \varepsilon_{\mathrm{val}}.
\end{equation}
\end{theorem}

\begin{proof}
Introduce the proper HJB operator
\[
F(\tau,x,r,p,q):=q-\tilde H(\tau,x,p)+\lambda r,
\]
so that \eqref{eq:forward_HJB_base_slack_offset_roi} is $F(\tau,x,W_\lambda,\nabla W_\lambda,\partial_\tau W_\lambda)=0$.

\medskip
\noindent\textbf{Step 1 (Adjust the PDE slack by a constant shift).}
Let $c_{\mathrm{pde}}:=\varepsilon_{\mathrm{pde}}/\lambda$ and define
\[
U:=\widehat W-c_{\mathrm{pde}},
\qquad
V:=\widehat W+c_{\mathrm{pde}}.
\]

\smallskip
\noindent\emph{Claim 1: $U$ is a viscosity subsolution of $F\le 0$ on $(0,T]\times X$.}
Let $\psi\in C^1$ and suppose $U-\psi$ has a local maximum at $(\tau_0,x_0)\in(0,T]\times X$.
Then $\widehat W-(\psi+c_{\mathrm{pde}})$ has a local maximum at the same point since
$U-\psi=\widehat W-(\psi+c_{\mathrm{pde}})$.
Apply \eqref{eq:pde_slack_sub_offset_roi} with $\phi:=\psi+c_{\mathrm{pde}}$ to obtain
\[
F(\tau_0,x_0,\psi(\tau_0,x_0)+c_{\mathrm{pde}},\nabla\psi(\tau_0,x_0),\partial_\tau\psi(\tau_0,x_0))
\le \varepsilon_{\mathrm{pde}}.
\]
Because $F$ is affine in $r$ with slope $\lambda$ and $c_{\mathrm{pde}}=\varepsilon_{\mathrm{pde}}/\lambda$, this rearranges to
\[
F(\tau_0,x_0,\psi(\tau_0,x_0),\nabla\psi(\tau_0,x_0),\partial_\tau\psi(\tau_0,x_0))\le 0.
\]
Thus $U$ is a viscosity subsolution of $F\le 0$.

\smallskip
\noindent\emph{Claim 2: $V$ is a viscosity supersolution of $F\ge 0$ on $(0,T]\times X$.}
Let $\psi\in C^1$ and suppose $V-\psi$ has a local minimum at $(\tau_0,x_0)\in(0,T]\times X$.
Then $\widehat W-(\psi-c_{\mathrm{pde}})$ has a local minimum at the same point since
$V-\psi=\widehat W-(\psi-c_{\mathrm{pde}})$.
Apply \eqref{eq:pde_slack_sup_offset_roi} with $\phi:=\psi-c_{\mathrm{pde}}$ to obtain
\[
F(\tau_0,x_0,\psi(\tau_0,x_0),\nabla\psi(\tau_0,x_0),\partial_\tau\psi(\tau_0,x_0))\ge 0.
\]
Hence $V$ is a viscosity supersolution of $F\ge 0$.

\medskip
\noindent\textbf{Step 2 (enforce the initial ordering at $\tau=0$).}
By \eqref{eq:init_mismatch_offset_roi} and $\varepsilon_{\mathrm{val}}\ge \varepsilon_0$,
\begin{multline*}
    \underline W(0,x)=\widehat W(0,x)-\varepsilon_{\mathrm{val}}\le 0=W_\lambda(0,x),
\qquad \\
\overline W(0,x)=\widehat W(0,x)+\varepsilon_{\mathrm{val}}\ge 0=W_\lambda(0,x),
\quad \forall x\in X
\end{multline*}
Moreover, since $\varepsilon_{\mathrm{val}}\ge c_{\mathrm{pde}}$, the functions
$\underline W = U-(\varepsilon_{\mathrm{val}}-c_{\mathrm{pde}})$ and
$\overline W = V+(\varepsilon_{\mathrm{val}}-c_{\mathrm{pde}})$ remain, respectively, a viscosity subsolution of $F\le 0$
and a viscosity supersolution of $F\ge 0$ on $(0,T]\times X$ (because shifting by a constant preserves the inequality direction for a proper operator).

\medskip
\noindent\textbf{Step 3 (comparison).}
By Claims 1--2, $\underline W$ is a viscosity subsolution of $F\le 0$ and $\overline W$ is a viscosity supersolution of $F\ge 0$ on $(0,T]\times X$,
and Step~2 provides the ordering at $\tau=0$ on $X$. By the comparison principle for proper, continuous HJB operators under the standing assumptions and $\lambda>0$,
\[
\underline W \le W_\lambda \le \overline W \quad \text{on } \mathcal D,
\]
which proves \eqref{eq:bracket_conclusion_offset_roi}. The sup-norm bound \eqref{eq:supnorm_bound_from_slack_offset_roi} follows immediately.
\end{proof}

\begin{theorem}[Reachable-set bracketing from value error]\label{thm:reach_bracket}
Assume \Cref{ass:S0,ass:S2}. Then, for every $\tau\in[0,T]$,
\begin{multline}
\label{eq:reach_brackets_clean}
\{x:\,\widehat W(\tau,x)<-\varepsilon_{\mathrm{val}}\}
\ \subseteq\ \{x:\,W_\lambda(\tau,x)<0\}\ \\
\subseteq\ \{x:\,\widehat W(\tau,x)\le +\varepsilon_{\mathrm{val}}\},
\end{multline}
where $\{x:\,W_\lambda(\tau,x)<0\}$ equals the strict backward-reachable set at horizon $\tau$.
\end{theorem}
\begin{proof}
Fix $\tau\in[0,T]$ and assume the uniform value bound
$\|\widehat W-W_\lambda\|_{L^\infty(\{\tau\}\times\R^n)}\le \varepsilon_{\mathrm{val}}$.
Under \Cref{ass:S0,ass:S2}, strict backward reachability satisfies
$\mathsf{BRT}(\tau)=\{x:\,W_\lambda(\tau,x)<0\}$. See Sec.~3.3, Prop.~(Strict Sign) in \cite{SolankiHJRL2025}.

\emph{(Left inclusion).} Let $x$ satisfy $\widehat W(\tau,x)<-\varepsilon_{\mathrm{val}}$.
Then
\[
W_\lambda(\tau,x)\ \le\ \widehat W(\tau,x)+\varepsilon_{\mathrm{val}}\ <\ 0,
\]
so $x\in\mathsf{BRT}(\tau)$.

\emph{(Right inclusion).} Let $x\in\mathsf{BRT}(\tau)$, i.e.\ $W_\lambda(\tau,x)<0$.
Then
\[
\widehat W(\tau,x)\ \le\ W_\lambda(\tau,x)+\varepsilon_{\mathrm{val}}\ \le\ \varepsilon_{\mathrm{val}},
\]
so $x\in \{y:\,\widehat W(\tau,y)\le \varepsilon_{\mathrm{val}}\}$.

Combining the two gives
\begin{equation}
\begin{aligned}
    \{x:\,\widehat W(\tau,x)<-\varepsilon_{\mathrm{val}}\}
&\subseteq \{x:\,W_\lambda(\tau,x)<0\}  \\
&\subseteq \{x:\,\widehat W(\tau,x)\le \varepsilon_{\mathrm{val}}\} .
\end{aligned}
\end{equation}
\end{proof}

\section{SMT-Based Certification of Uniform Value Error on a Region of Interest}
\label{sec:smt_cert}

To certify a learned approximation $\widehat W$ on a prescribed compact region of interest (ROI)
$\mathcal D\subseteq[0,T]\times\R^n$, we seek a \emph{uniform value-function error bound}
\begin{equation}\label{eq:goal_eps_val}
\|\widehat W - W_\lambda\|_{L^\infty(\mathcal D)}\ \le\ \varepsilon_{\mathrm{val}},
\end{equation}
where $W_\lambda$ denotes the (unknown) true discounted travel value solving the forward HJB
\Cref{eq:HJB_forward_tau}. Since $W_\lambda$ is not available in closed form, we certify
$\varepsilon_{\mathrm{val}}$ indirectly via \emph{residual} certificates that are amenable to SMT.

In this section we assume that the relevant optimal control (or disturbance) admits a \emph{closed-form selector}
so that (i) the Hamiltonian $\tilde H$ can be evaluated without an explicit $\inf_{u\in U}$,
and (ii) the one-step Bellman minimization can be evaluated without an explicit $\inf_{a(\cdot)\in\mathcal A_\sigma}$.
Under these assumptions, both Route~A and Route~B reduce to checking the absence of a counterexample point
$(\tau,x)\in\mathcal D$ for a fully explicit scalar inequality.

Specifically, we use two routes:
\begin{enumerate}
\item[\textbf{A)}] a discrete-time \emph{Bellman-operator residual} bound on $\mathcal D_\sigma:=[\sigma,T]\times X$
which, by contraction of $T_{\sigma,\lambda}$ (\Cref{lem:contr_value_err}), implies \eqref{eq:goal_eps_val}
(on $\mathcal D_\sigma$) with $\varepsilon_{\mathrm{val}}=\varsigma_{\mathrm{op}}/(1-e^{-\lambda\sigma})$; and
\item[\textbf{B)}] a continuous-time \emph{HJB (PDE) slack} bound on $\mathcal D=[0,T]\times X$ which,
via \Cref{thm:slack_implies_bracket_offset_roi}, implies \eqref{eq:goal_eps_val} with
$\varepsilon_{\mathrm{val}}=\max\{\varepsilon_{\mathrm{pde}}/\lambda,\varepsilon_0\}$.
\end{enumerate}

In both routes, certification reduces to proving the \emph{absence} of a counterexample.
Concretely, to certify $\|g\|_{L^\infty(\mathcal D)}\le \eta$ for a scalar expression $g(\tau,x)$, we ask whether there exists
$(\tau,x)\in\mathcal D$ such that $|g(\tau,x)|>\eta$. This is naturally expressed as an existential satisfiability query
and can be encoded as a quantifier-free nonlinear real-arithmetic formula (with ODE semantics where needed) for SMT solving.

\begin{remark}[Closed-form selectors after learning $\widehat W$]
\label{rem:closedform_not_unrealistic}
The assumption that the minimization in Route~B (and, in structured cases, Route~A) admits a closed-form
selector is often realistic. Once a differentiable approximation $\widehat W$ is learned, its gradient
$p(\tau,x):=\nabla_x\widehat W(\tau,x)$ is available. The HJB
structure then induces an \emph{instantaneous} minimizing control through the Hamiltonian integrand
\begin{multline*}
    \mathcal H(\tau,x,p;u):=h(T-\tau,x,u)+p\cdot f(x,u),
\qquad \\
u^\star(\tau,x)\in\arg\min_{u\in U}\mathcal H\!\big(\tau,x,\nabla_x\widehat W(\tau,x);u\big)
\end{multline*}
Under \Cref{ass:A1,ass:A4,ass:A5}, the argmin is nonempty for each $(\tau,x)$ (compact $U$ and continuity in $u$),
so $\tilde H(\tau,x,\nabla_x\widehat W(\tau,x))$ can be evaluated by substituting $u^\star(\tau,x)$.
For common control-affine dynamics $f(x,u)=f_0(x)+G(x)u$ with convex compact $U$ (boxes/balls/polytopes) and
costs $h$ that are linear or convex in $u$, the selector $u^\star$ is often available in closed form
(e.g., saturation/projection or bang bang at extreme points).

In contrast, the one-step Bellman minimization in Route~A is an open-loop optimization over $a(\cdot)\in\mathcal A_\sigma$
with a terminal term $\widehat W(\tau-\sigma,y(\sigma))$. In general, the instantaneous selector $u^\star(\tau,x)$
provides an admissible candidate control for this problem, and it coincides with the one-step minimizer only under
additional structure (e.g., special convexity/linearity conditions or when $\widehat W$ is itself a fixed point).
\end{remark}

\subsection{Two SMT routes to a certified value bound}
\label{subsec:smt_routes}

We present two complementary SMT-amenable routes that deliver a certified uniform error
$\varepsilon_{\mathrm{val}}$ on a compact ROI $\mathcal D$.

\subsubsection{Route A: Bellman-operator residual on an invariant ROI}
\label{subsec:routeA_watertight}

We work with a compact state ROI $X\subset\R^n$ and define the spacetime ROI
\[
\mathcal D := [0,T]\times X,
\qquad
\mathcal D_\sigma := [\sigma,T]\times X .
\]

\paragraph{One-step closure (forward invariance over $\sigma$).}
To apply the one-step Bellman operator on $\mathcal D_\sigma$ while keeping all evaluations inside $X$,
we require the following closure property.

\begin{assumption}[One-step forward invariance of $X$ over $\sigma$]
\label{ass:one_step_invariance}
For every $x\in X$ and every admissible intra-step control $a(\cdot)$ on $[0,\sigma]$,
the corresponding trajectory $y(\cdot)$ solving
$y'(r)=f(y(r),a(r))$, $y(0)=x$ satisfies $y(\sigma)\in X$.
\end{assumption}

\paragraph{Closed-form one-step minimizer.}
Recall the one-step cost and Bellman operator from \Cref{subsec:discrete-time-ttg}:
\begin{multline*}
    c(\tau,x,a):=\int_{0}^{\sigma} e^{-\lambda r}\,h\!\big(T-\tau+r,\,y(r),\,a(r)\big)\,dr,
\qquad \\
(T_{\sigma,\lambda}\Psi)(\tau,x)
:= \inf_{a\in\mathcal A_\sigma}\Big\{c(\tau,x,a)+\gamma\,\Psi(\tau-\sigma,y(\sigma))\Big\}
\end{multline*}
In this version, we assume that for $\Psi=\widehat W$ the infimum is attained by a known closed-form input.

\begin{assumption}[Closed-form minimizer for the one-step problem]
\label{ass:closedform_bellman_min}
There exists a measurable selector $a^\star(\tau,x):[0,\sigma]\to U$ such that for every $(\tau,x)\in\mathcal D_\sigma$,
letting $y^\star(\cdot)$ solve $y^{\star\prime}(r)=f(y^\star(r),a^\star(\tau,x)(r))$, $y^\star(0)=x$, we have
\begin{equation}\label{eq:ass_bellman_attained}
(T_{\sigma,\lambda}\widehat W)(\tau,x)
=
c(\tau,x,a^\star(\tau,x))+\gamma\,\widehat W(\tau-\sigma,y^\star(\sigma)).
\end{equation}
\end{assumption}

\paragraph{Residual to certify.}
Define the \emph{explicit} one-step residual (no $\inf$ remaining)
\begin{equation}\label{eq:routeA_residual_explicit}
R_{\mathrm{op}}(\tau,x)
:=
\Big|
c(\tau,x,a^\star(\tau,x))
+\gamma\,\widehat W(\tau-\sigma,y^\star(\sigma))
-\widehat W(\tau,x)
\Big|.
\end{equation}
By \Cref{ass:closedform_bellman_min}, this equals $\big|(T_{\sigma,\lambda}\widehat W)(\tau,x)-\widehat W(\tau,x)\big|$ on $\mathcal D_\sigma$.

\paragraph{Certified value error from a certified operator residual.}
If an SMT certificate yields $R_{\mathrm{op}}(\tau,x)\le \varsigma$ for all $(\tau,x)\in\mathcal D_\sigma$, then
\[
\|T_{\sigma,\lambda}\widehat W-\widehat W\|_{L^\infty(\mathcal D_\sigma)}\le \varsigma,
\]
and by \Cref{lem:contr_value_err},
\begin{equation}\label{eq:routeA_val_bound_watertight}
\|\widehat W-W_\lambda\|_{L^\infty(\mathcal D_\sigma)}
\ \le\ \frac{\varsigma}{1-\gamma}
\ =\ \frac{\varsigma}{1-e^{-\lambda\sigma}}
\ =:\ \varepsilon_{\mathrm{val}}.
\end{equation}

\paragraph{SMT obligation (counterexample query).}
To certify $R_{\mathrm{op}}(\tau,x)\le \varsigma$ on $\mathcal D_\sigma$, partition $\mathcal D_\sigma$ into cells $\{C_i\}_{i=1}^N$
and, for each cell, ask whether there exists a violating point $(\tau,x)$:
\begin{equation}\label{eq:smt_routeA_watertight}
\exists (\tau,x)\in C_i:\quad R_{\mathrm{op}}(\tau,x)>\rho_{C_i},
\end{equation}
together with the ODE constraint defining $y^\star(\cdot)$ under $a^\star(\tau,x)$.
A $\delta$-\textsc{UNSAT} result for \eqref{eq:smt_routeA_watertight} implies
$\sup_{(\tau,x)\in C_i}R_{\mathrm{op}}(\tau,x)\le \rho_{C_i}$.
Setting $\varsigma_{\mathrm{op}}(\mathcal D_\sigma):=\max_i\rho_{C_i}$ then yields \eqref{eq:routeA_val_bound_watertight}.

\begin{remark}[Fallback when a closed-form one-step minimizer is unavailable]
If \Cref{ass:closedform_bellman_min} is not available, one can revert to the conservative ``strong residual''
formulation that introduces $\sup_{a\in\mathcal A_\sigma}|\cdot|$,
which yields an $\exists(\tau,x)\exists a$ SMT query but remains sound.
\end{remark}

\subsubsection{Route B: HJB (PDE) slack on an ROI}
\label{subsec:routeB_watertight}

Route~B certifies $\varepsilon_{\mathrm{val}}$ by bounding the forward HJB residual of $\widehat W$ on the ROI
$\mathcal D=[0,T]\times X$ directly, without introducing the discrete-time Bellman operator.

\paragraph{Closed-form Hamiltonian.}
Recall
\[
\tilde H(\tau,x,p)=\inf_{u\in U}\Big\{h(T-\tau,x,u)+p\cdot f(x,u)\Big\}.
\]
In this version, we assume that the infimum is attained by a known closed-form selector.

\begin{assumption}[Closed-form minimizer for the Hamiltonian]
\label{ass:closedform_hamiltonian_min}
There exists a measurable selector $u^\star(\tau,x,p)\in U$ such that for all $(\tau,x,p)$,
\begin{multline}\label{eq:ass_H_attained}
    \tilde H(\tau,x,p)
= 
h(T-\tau,x,u^\star(\tau,x,p))+ \\ p\cdot f(x,u^\star(\tau,x,p))
\end{multline}
\end{assumption}

\paragraph{Residual to certify.}
With $\tilde H$ evaluated via \eqref{eq:ass_H_attained}, define the forward HJB residual
\begin{multline}\label{eq:def_HJB_residual_forward}
R_{\mathrm{HJB}}(\widehat W)(\tau,x)
:= \partial_\tau\widehat W(\tau,x)-\tilde H\!\big(\tau,x,\nabla_x\widehat W(\tau,x)\big) \\
+\lambda\,\widehat W(\tau,x).
\end{multline}
We certify the uniform slack bound
\begin{equation}\label{eq:routeB_target_residual}
\|R_{\mathrm{HJB}}(\widehat W)\|_{L^\infty((0,T]\times X)} \le \varepsilon_{\mathrm{pde}},
\end{equation}
together with an initial mismatch bound
\begin{equation}\label{eq:routeB_target_init}
\varepsilon_0 \ \ge\ \sup_{x\in X}|\widehat W(0,x)|.
\end{equation}
Then \Cref{thm:slack_implies_bracket_offset_roi} yields
\begin{equation}\label{eq:routeB_epsval}
\|\widehat W-W_\lambda\|_{L^\infty(\mathcal D)}
\ \le\ \varepsilon_{\mathrm{val}}
\ :=\ \max\Big\{\frac{\varepsilon_{\mathrm{pde}}}{\lambda},\ \varepsilon_0\Big\}.
\end{equation}

\paragraph{SMT obligations (counterexample queries).}
Partition $(0,T]\times X$ into cells $\{C_i\}_{i=1}^N$ and choose per-cell bounds $\rho^{\mathrm{pde}}_{C_i}$.
For each cell, certify the absence of a violating point:
\begin{equation}\label{eq:smt_routeB_cell}
\neg\exists (\tau,x)\in C_i:\quad \big|R_{\mathrm{HJB}}(\widehat W)(\tau,x)\big|>\rho^{\mathrm{pde}}_{C_i}.
\end{equation}
A $\delta$-\textsc{UNSAT} result implies $\sup_{(\tau,x)\in C_i}|R_{\mathrm{HJB}}(\widehat W)(\tau,x)|\le \rho^{\mathrm{pde}}_{C_i}$.
Define $\varepsilon_{\mathrm{pde}}:=\max_i\rho^{\mathrm{pde}}_{C_i}$.
Similarly, certify the initial mismatch on $X$ by
\begin{equation}\label{eq:smt_routeB_init}
\neg\exists x\in X:\quad |\widehat W(0,x)|>\rho^0_X,
\end{equation}
and set $\varepsilon_0:=\rho^0_X$.
Finally, set $\varepsilon_{\mathrm{val}}$ as in \eqref{eq:routeB_epsval} and invoke
\Cref{thm:slack_implies_bracket_offset_roi} to conclude the certified enclosure
$\widehat W-\varepsilon_{\mathrm{val}}\le W_\lambda\le \widehat W+\varepsilon_{\mathrm{val}}$ on $\mathcal D$.

\begin{remark}[Fallback when a closed-form Hamiltonian is unavailable]
If \Cref{ass:closedform_hamiltonian_min} is not available, one can revert to the conservative
formulation, which introduces an $\exists(\tau,x)\exists u$ SMT query but remains sound.
\end{remark}

\begin{remark}[Infinite-horizon (stationary) specialization]
The certification pipeline applies verbatim to the infinite-horizon discounted problem.
In particular, when the dynamics and running cost are time-invariant and one seeks the
stationary discounted value $W_\lambda(x)$, the HJB reduces to
\[
\lambda W_\lambda(x) + \min_{u\in U}\Big\{ h(x,u) + \nabla W_\lambda(x)\cdot f(x,u)\Big\} = 0
\quad \text{on } X,
\]
(with the appropriate target/obstacle sign convention).
In this case the learned candidate depends only on $x$ and the PDE-residual route
(Route~B) certifies a uniform bound $\|W_c-W_\lambda\|_{L^\infty(X)}\le \varepsilon_{\mathrm{val}}$
from a certified residual bound and boundary mismatch exactly as in Theorem~3.
Thus, although our exposition is written on a spacetime ROI $[0,T]\times X$, the stationary
discounted setting is recovered by specializing to $\partial_\tau W\equiv 0$ and replacing
$[0,T]\times X$ by $X$.
\end{remark}


\subsection{SMT-backed verification via CEGIS}
\label{subsec:smt_cegis_unified}

We couple learning and certification in a counterexample-guided inductive synthesis (CEGIS) loop~\cite{abate2018counterexample}.
Starting from a parametric approximation $\widehat W(\,\cdot\,;\theta)$, we alternate:
(i) training on sampled data; and
(ii) parallel SMT calls on each cell $C_i$ to search for counterexamples to the desired residual bounds.

\paragraph{Route A (operator residual).}
For Route~A (\Cref{subsec:routeA_watertight}), the SMT obligation is the absence of a violating triple
$(\tau,x,a)$ in each cell, as encoded in \eqref{eq:smt_routeA_watertight}.
A $\delta$-\textsc{SAT} result returns a witness $(\tau^\star,x^\star,a^\star)$, which is used to enrich the training set near the violating region.
If all cells return $\delta$-\textsc{UNSAT}, then the certified strong-residual bound implies the value-error bound
\eqref{eq:routeA_val_bound_watertight}, yielding $\varepsilon_{\mathrm{val}}$ on $\mathcal D_\sigma$.

\paragraph{Route B (PDE slack).}
For Route~B (\Cref{subsec:routeB_watertight}), certification proceeds by ruling out counterexamples to the uniform band condition:
for each cell $C_i$, the SMT query \eqref{eq:smt_routeB_cell} searches for a violating triple $(\tau,x,u)$.
In addition, the initial mismatch on $X$ is certified by \eqref{eq:smt_routeB_init}.
If all cells are $\delta$-\textsc{UNSAT}, we obtain certified bounds $\varepsilon_{\mathrm{pde}}$ and $\varepsilon_0$, which define
$\varepsilon_{\mathrm{val}}$ via \eqref{eq:routeB_epsval}. Then \Cref{thm:slack_implies_bracket_offset_roi}
yields the uniform enclosure $\widehat W-\varepsilon_{\mathrm{val}}\le W_\lambda\le \widehat W+\varepsilon_{\mathrm{val}}$ on $\mathcal D$.

\paragraph{Certified reachable-set brackets.}
Once $\varepsilon_{\mathrm{val}}$ is certified (by either route), the strict reachability brackets follow from
\Cref{thm:reach_bracket}, i.e., \eqref{eq:reach_brackets_clean} holds for every $\tau\in[0,T]$.

Backends such as dReal support nonlinear real arithmetic with polynomial/transcendental nonlinearities and provide
$\delta$-complete reasoning for ODE constraints~\cite{gao2013dreal}.

\section{Experiment: Double Integrator}
\label{sec:experiment}
\todo{What about the enclosures?}

In this section, we demonstrate how the proposed framework can be used
to \emph{formally certify} a value function learned via reinforcement
learning.
The purpose of this experiment is not to benchmark learning performance,
but to provide a concrete, end-to-end instantiation of the certification
pipeline developed in \Cref{sec:hjb_add_eps,sec:rl_counterpart_eps,sec:smt_cert}. 

As a testbed, we consider the time-invariant double integrator,
a canonical control system for which Hamilton--Jacobi reachability
structure is well understood. The dynamics are given by \cref{eq:di_dyn}.
\begin{equation}
\dot x_{1} = x_{2},\qquad \dot x_{2} = u,\qquad u\in\mathcal U:=[-1,1] .
\label{eq:di_dyn}
\end{equation}

We take a discount rate $\lambda=1 \ s^{-1}$, and the running cost is chosen as follows
\begin{equation}
h(x) = \begin{cases}
-\alpha (r_{g}-|x|), & |x|\leq r_g,\\
0, & \text{otherwise},
\end{cases}
\qquad (\alpha=1,\ r_g=0.5),
\label{eq:goal_band_cost}
\end{equation}
where $\alpha$ and $r_g$ denote the cost scale and the target radius, respectively.
The choice of the travel cost as in \cref{eq:goal_band_cost} ensures that the value is negative inside the target tube and zero elsewhere, thereby
recovering strict backward-reachability semantics by sign. 

\subsection{Value function learning setup} We learn a stationary discounted value function $\widehat W : \mathbb{R}^2 \to \mathbb{R}$ using a reinforcement-learning procedure aligned with the discounted time-to-go (TTG) Bellman operator introduced in \cref{sec:rl_counterpart_eps}. The learned value serves as a candidate approximation to the viscosity solution of the discounted Hamilton--Jacobi--Bellman equation and is subsequently subjected to formal certification. 

Training combines temporal-difference (TD) learning with an explicit discounted HJB residual penalty and a semi-Lagrangian discretization of the Bellman map. Learning is restricted to a compact, axis-aligned region of interest $\mathcal{D} \subset \mathbb{R}^2$.

The value function is parameterized by a multi-layer perceptron with sinusoidal activation functions (SIREN), and the sole input to the network is the state.

All architectural choices, optimization hyperparameters, and training settings are summarized in \cref{tab:hyperparameters}.

\todo{Add dotted lines as sub-separators in the table.}
\todo{Is SIREN a type of nertwork or strictly the activation?}

\begin{table}[t]
\centering
\caption{Key hyperparameters and configuration parameters of the experiment.}
\label{tab:hyperparameters}
\begin{tabular}{ll}
\hline
\textbf{Problem and Reachability Setup} & \\ \hline
Dynamics & Double integrator (\cref{eq:di_dyn}) \\
Running cost & \cref{eq:goal_band_cost} \\
Target radius $r_g$ & $0.5$ \\
Cost scale $\alpha$ & $1$ \\ 
Control set $\mathcal{U}$ & $[-1,1]$ \\ \hline

\textbf{RL / Value Iteration} & \\ \hline
Time step $\Delta \tau$ & $0.05 \ s$ \\
Discount rate $\lambda$ & $1.0 \ s^{-1}$\\
Discount factor $\gamma$ & $e^{-\lambda \Delta \tau} \approx 0.9512294245$ \\
TD discretization & Semi-Lagrangian (2nd order) \\
Replay buffer size & $4 \times 10^{6}$ \\
Batch size & $163840$ \\
Training iterations & $10^{5}$ \\
Target update rate $\tau$ & $5 \times 10^{-3}$ \\ 
Loss weights & $w_{\mathrm{TD}} = 1,\; w_{\mathrm{HJB}} = 1,\; w_{\mathrm{Sob}} = 0$ \\ \hline

\textbf{Neural Network Architecture} & \\ \hline
Network type & SIREN MLP \\
Hidden layers & $(40, 40)$ \\
Activation & $\sin(\cdot)$ \\
Frequency parameter $w_0$  & $30$ \\ \hline

\textbf{Certification and ROI} & \\ \hline
Region of interest (ROI) & $[-2.5,\,2.5]^2$ \\
Certified residual tolerance $\varepsilon$ & $0.1$ \\
SMT precision $\delta$ & $10^{-8}$ \\
Certification backend & dReal \\ \hline
\end{tabular}
\end{table}


\subsection{From RL model to symbolic expression}
After training, we export the learned value function $\widehat W$ to an exact symbolic expression $\widehat W_{\mathrm{sym}}(x)$ by composing the learned linear layers and sinusoidal activation functions (without approximation). This yields a single closed-form SymPy expression representing the value function over the state space \cite{meurer2017sympy}.

This symbolic representation enables analytic differentiation and exact evaluation of the discounted Hamilton--Jacobi--Bellman residual. Crucially, it allows the certification problem to be expressed as a quantifier-free first-order formula over real variables, making it amenable to satisfiability modulo theories (SMT) solving. In this way, the symbolic export serves as the interface between the learned RL model and the subsequent formal verification procedure.

\subsection{PDE residual and SMT specification}
\todo{Should we use $\varepsilon_{\text{val}}$ instead?}
\todo{Add what is equation 75 based on.}
\todo{Should we use widehat?}

To formally certify the learned value function, we evaluate its violation of the discounted Hamilton--Jacobi--Bellman (HJB) equation, specialized to the stationary double-integrator case by setting $\partial_\tau W_\lambda(\tau,x)$ in \cref{eq:HJB_forward_tau} to $0$. This leads to
\begin{equation}
\lambda W(x_{1},x_{2})=h(x)\;+\;\min_{u\in\mathcal U}\!\big(\nabla W(x_{1},x_{2})\cdot[x_{2},\ u]\big) .
\label{eq:hjb_stationary}
\end{equation}
Thus, the residual of a candidate $\widehat W$ is
\begin{multline}
\mathcal R[\widehat W](x_{1},x_{2}):=\lambda \widehat W(x_{1},x_{2})-h(x) \\ -\min\{\,\widehat W_{x_{1}} x_{2}+\widehat W_{x_{2}} u_{\min},\ \widehat W_{x_{1}} x_{2}+ \widehat W_{x_{2}} u_{\max}\,\}.
\label{eq:residual}
\end{multline}
Because the Hamiltonian is affine in $u$, the minimum over $[u_{\min},u_{\max}]$ is attained at an endpoint.
We verify a \emph{uniform} bound on the HJB residual over $\mathcal D$ by searching, with dReal, for a counterexample to
\begin{equation}
\bigl|\mathcal R[\widehat W](x_{1},x_{2})\bigr|\le \varepsilon_{\text{val}}
\qquad\forall(x_{1},x_{2})\in\mathcal D.
\label{eq:uniform_bound}
\end{equation}
The SMT formula encodes the ROI box, the goal–band split, the endpoint selector for $\min$, and the two–sided inequality in \cref{eq:uniform_bound}. We use dReal’s $\delta$–decision procedure and report $\delta$–UNSAT as certification that no counterexample exists within precision~$\delta$.

\subsection{Results}

\textbf{Qualitative: } The learned value is strictly negative in the target band and approaches~$0$ outside; the induced bang–bang policy splits $\mathcal D$ into the expected two regions.

\textbf{Numerical cross-check: } We also solved the discounted stationary HJB on the same ROI via a semi-Lagrangian fixed-point scheme (explicit in drift, endpoint minimization). \Cref{fig:pde-vs-nn} compares the numerical solution $W_{\text{num}}$ with the network $W_{\text{nn}}$ and the pointwise difference $\Delta=W_{\text{num}}-W_{\text{nn}}$. The two fields are visually consistent; discrepancies concentrate near the switching curve. 
\todo{Add justification for the above observation.}
Moreover, the color scale shows $|\Delta|=\mathcal{O}(10^{-2})$, which is \emph{consistent} with the SMT-certified residual tolerance $\varepsilon_{\text{val}}=0.1$, while $\Delta$ is a value-error proxy (not the residual), this provides a quick visual sanity check of the certificate.

\begin{figure*}[t]
  \centering
  \includegraphics[width=\textwidth]{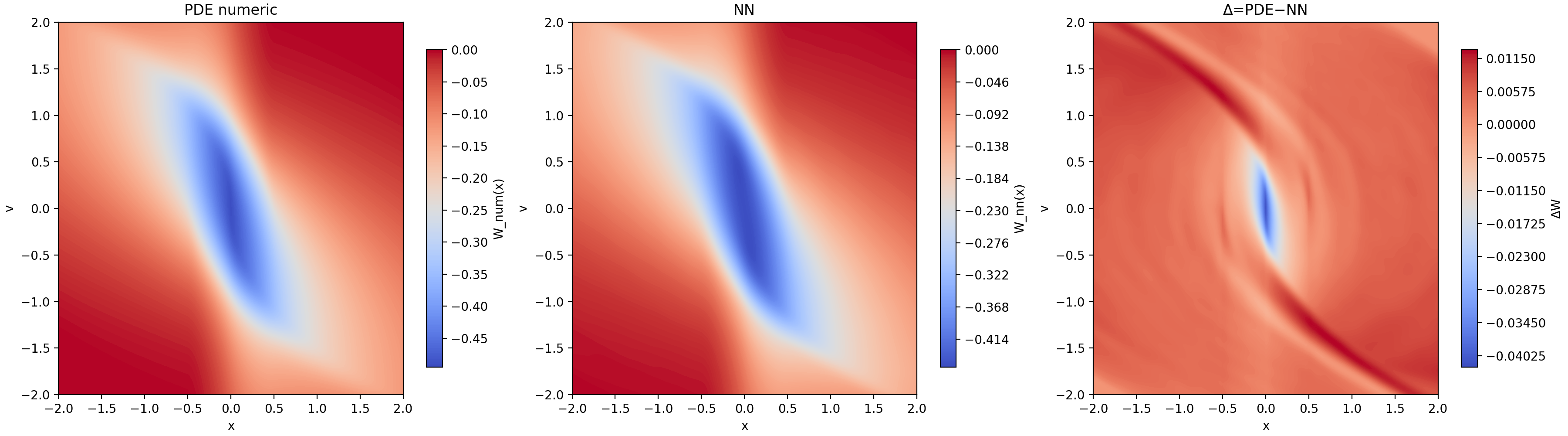}
  \caption{PDE numeric solution, NN prediction, and their difference $\Delta=W_{\text{num}}-W_{\text{nn}}$ on the same grid. Errors are small and localized near the bang–bang switching locus and the ROI boundary; the observed magnitude ($\approx 10^{-2}$) is consistent with the SMT-certified residual tolerance $\varepsilon=0.1$.}
  \label{fig:pde-vs-nn}
\end{figure*}

\textbf{SMT certification:} Using the exported $W$ and analytic gradients, dReal returns
\[
\texttt{status}=\texttt{UNSAT},\quad \delta=10^{-8},\quad \varepsilon_{\text{val}}=0.1 .
\]
Thus, up to $\delta$–weakening, there is \emph{no} $(x,v)$ in the ROI for which
\begin{multline*}
    \bigl|\lambda W(x_{1},x_{2})-h(x) \\ -\min\{\,W_{x_{1}} x_{2}+W_{x_{2}} u_{\min},\ W_{x_{1}} x_{2}+W_{x_{2}} u_{\max}\,\}\bigr|>\,0.1 .
\end{multline*}
This yields a certified uniform PDE residual bound of $\varepsilon_{\text{val}}=0.1$ on the entire box.
\smallskip

\noindent\textit{Remark (tightening $\varepsilon_{\text{val}}$).}
The $\varepsilon_{\text{val}}$ reported here is intentionally conservative and set by the user in the SMT query. Smaller $\varepsilon_{\text{val}}$ values are attainable with richer function classes, longer training, or stronger PDE regularization; the same pipeline scales unchanged.

\subsection{Reproducibility}
\todo{Rename scripts and update this because currently DQN is not used.}

All experiments were executed headless in WSL\,2 (Ubuntu~22.04, Python~3.10, PyTorch~2.9), using the script \texttt{experiments/run\_double\_integrator\_dqn.py}. Symbolic export and the PDE SMT check are performed by \texttt{tools/symbolify\_rl\_model.py} and \texttt{tools/smt\_check\_pde\_dreal.py}, respectively. The ROI, discount, control bounds, and cost parameters used by dReal are read directly from the training run’s JSON metadata to ensure consistency.

\section{Conclusion}
\label{sec:conclusion}

This paper introduced a framework for certifying Hamilton Jacobi reachability that has been learned via reinforcement learning. The central insight is that an additive-offset to the value induces a constant offset in the discounted HJB equation. This identity allows certified bounds on Bellman or HJB residuals to be translated directly into provably correct inner and outer enclosures of backward-reachable sets.

We presented two SMT-based certification routes: one that is operator-based and another that is PDE-based. Both were cast as counterexample-driven satisfiability queries over a compact region of interest and compatible with CEGIS. The approach was validated on a double-integrator system, where an RL-learned neural value function was formally certified.

While counterexamples identified during certification have been used to heuristically refine or retrain the value function, the framework does not endow this process with convergence guarantees toward a certifiable solution. Reinforcement learning therefore serves to generate candidate value functions, while formal certification acts as the acceptance criterion for semantic correctness.

Future research directions include the development of learning algorithms with explicit convergence guarantees toward certifiable value functions, the reduction of conservatism in uniform certificates over regions of interest, and extensions of the framework to stochastic or time-varying optimal control problems. Improving the scalability of formal certification to higher-dimensional systems also remains an important challenge.

\todo{Add steps related to relaxing assumptions.}

\bibliographystyle{plain}              
\bibliography{references}   

\end{document}